\documentclass[12pt]{amsart}
\usepackage{amsmath,amssymb,amsthm}

\usepackage{fullpage}

\newtheorem{lem}{Lemma}[section]

\newtheorem{thm}{Theorem}[section]

\theoremstyle{definition}

\theoremstyle{remark}

\theoremstyle{remark}
\newtheorem{remark}{Remark}[section]
\numberwithin{equation}{section}

\newcommand{\set}[1]{\left\{#1\right\}}

\newcommand{\Z}{{\mathbb Z}}
\newcommand{\R}{{\mathbb R}}
\newcommand{\OO}{\mathcal{O}}

\begin{document}

\title{Nonlinear Schr\"odinger equations with strongly singular potentials}

\author{Jacopo Bellazzini}
\address{Dipartimento di Matematica Applicata ``U. Dini'', University of Pisa, via Buonarroti 1/c, 56127 Pisa, ITALY}%
\email{j.bellazzini@ing.unipi.it}%
\author{Claudio Bonanno}
\address{Dipartimento di Matematica Applicata ``U. Dini'', University of Pisa, via Buonarroti 1/c, 56127 Pisa, ITALY}%
\email{bonanno@mail.dm.unipi.it}%

\thanks{Authors are partially supported by M.I.U.R project PRIN2007 ``Variational and topological methods in the study of nonlinear phenomena''}

\maketitle

\begin{abstract}
In this paper we look for standing waves for nonlinear
Schr\"odinger equations 
$$
i\frac{\partial \psi }{\partial t}+\Delta \psi - g(|y|) \psi
-W^{\prime }(\left| \psi \right| )\frac{\psi }{\left| \psi \right|
}=0
$$
with cylindrically symmetric potentials $g$ vanishing at infinity and non-increasing,
and a $C^1$ nonlinear term satisfying weak assumptions. In
particular we show the existence of standing waves with
non-vanishing angular momentum with prescribed $L^2$ norm. The
solutions are obtained via a minimization argument, and the proof
is given for an abstract functional which presents lack of
compactness. As a particular case we prove the existence of standing waves with
non-vanishing angular momentum for the nonlinear hydrogen atom equation.
\end{abstract}

\section{Introduction} \label{intro}
In recent years much attention has been devoted to eigenvalue
problems for elliptic equations, mainly for applications to
nonlinear field equations, such as Schr\"odinger and Klein-Gordon
equations.

Let $N\ge 3$ and $k=2$. We write $x \in \R^N$ as $x=(y,z) \in \R^k
\times \R^{N-k}$. Consider the nonlinear Schr\"odinger equation
(NLS) in $\R^N$
\begin{equation} \label{pina}
i\frac{\partial \psi }{\partial t}+\Delta \psi - g(|y|) \psi
-W^{\prime }(\left| \psi \right| )\frac{\psi }{\left| \psi \right|
}=0
\end{equation}
with a potential $g$ vanishing at infinity and non-increasing, and
$W$ a nonlinear term of the kind studied in \cite{Beres-Lions}.
The existence of concentrated solutions of (\ref{pina}) 
can be obtained by looking for solutions of the form
\begin{equation} \label{ansatz}
\psi(t,x)=u(x)\, e^{i(\ell\, \theta(y)-\lambda t)}, \qquad u \geq
0,\ \lambda \in \R, \ \ell \in \Z
\end{equation}
where $\theta(y)$ is the angular variable in the plane
$(y_1,y_2)$. In particular if $\ell\neq 0$ these solutions have non-vanishing
angular momentum and are called vortices. With this ansatz, the NLS reduces to
\begin{equation} \label{nls-static}
-\triangle u + \left( \frac{\ell^2}{|y|^2} + g(|y|) \right)
u+W'(u)=\lambda u
\end{equation}
The problem of existence of vortices for nonlinear field equations
has been studied recently for $g\equiv 0$ in \cite{BeVi},
\cite{BBR-1}, \cite{BBR-2} and \cite{BaRo1}. In these papers,
solutions of (\ref{nls-static}) have been found as critical points
of a functional $J(u)$ constrained to the manifold of functions
with fixed $L^2$ norm. In the context of NLS, this constraint is
natural since the $L^2$ norm of a solution is an invariant of
motion. Moreover, it could be important to obtain points of
minimum to have orbital stability for the standing waves
(\ref{ansatz}). This is for example the case of solutions with
$\ell =0$. The main difficulty in this minimization problem is the
lack of compactness due to translations along the $z$-coordinates.
For this reason it was difficult to prove that the obtained
solution had the desired $L^2$ norm. This was solved in
\cite{BaRo1} under further assumptions on $W$.

In this paper, letting $2\le k \le N$, we consider the general
eigenvalue problem
\begin{equation} \label{ellittica}
-\triangle u + V(|y|) u+W'(u)=\lambda u
\end{equation}
where $V(s): \R^+ \to \R^+$ is assumed to be vanishing at infinity
and non-increasing, and the nonlinear term $W$ is of the kind
studied in \cite{Beres-Lions}. For precise assumptions see Section
\ref{problema-ellittico}. In particular (\ref{nls-static}) is of this
form with $V(|y|)=\left( \frac{\ell^2}{|y|^2} + g(|y|) \right)$ where
$g(|y|)$ can be strongly singular as $|y|^{-\alpha}$ with $\alpha >0$.
We prove the existence of solutions of
(\ref{ellittica}) with any desired $L^2$ norm large enough. These
solutions are obtained solving the minimization problem of the
functional
$$
J(u):=\int_{\R^N} \left( \frac 12 |\nabla u|^2 +\frac 12 V(|y|)\,
u^2+W(u) \right)dx
$$
restricted to cylindrically symmetric functions and constrained to
the manifold
$$
B_{\rho} := \set{ \int_{\R^N} u^2\, dx= \rho^2}
$$
The eigenvalue $\lambda$ is found as the Lagrange multiplier of
the minimization problem.

As far as we know, the only existence result for elliptic
equations with singular potentials of the form $V(|y|) \sim
|y|^{-\alpha}$ with $\alpha \not=2$ is contained in \cite{Ba2},
where it is considered the case $V(|y|) = |y|^{-\alpha}$ with
$\lambda=0$, and results depend on the relation between $\alpha$
and the growth conditions of $W$.

In Section \ref{problema-ellittico} we introduce the problem in
details. The proof of the existence of the constrained point of
minimum for $J$ is given in Section \ref{proof}. We first
introduce an abstract minimization problem for functionals of the
form
\begin{equation*}
I(u):=\left( \frac 12 \| u \|^2 + T(u) \right)
\end{equation*}
where $\| u \|$ is a suitable norm for functions in $H^{1}$ and $T$ is a real operator.
Under some weak assumptions on the behaviour of $T$ along
minimizing sequences, we prove in Theorem \ref{main-abs} the
existence of a point of minimum of $I$ constrained to the manifold
$B_\rho$. In particular we obtain strong convergence in $H^1$ for
any minimizing sequence. This approach has been inspired by
\cite{BelVi}, where in the case of non-singular
potentials and nonautonomous power-like nonlinear terms it was obtained
the orbital stability for standing waves of NLS.

In the application to $J$, we have $T(u):=\int W(u)dx$. The main
difficulty in dealing with such term is the lack of compactness on
the space of cylindrically symmetric functions. The idea of the
proof is the following: first we obtain an a-priori estimate to
guarantee the existence of a weak limit $\bar u$; second we prove
that $||\bar u||_{L^2}\neq 0$ (by means of a compactness lemma
contained in \cite{ESLI83}); as last step, by using the abstract
Theorem \ref{main-abs}, we show that $||\bar u||_{L^2}=\rho$ and
that $u_n\rightarrow \bar u$ strongly, and therefore $\bar u$ is
solution of (\ref{ellittica}).

Let us point out that the idea of the proof of the abstract
Theorem \ref{main-abs} can be applied in the case when $V\equiv
0$, i.e when the problem does not contain the singular term. Under
the same assumptions, we obtain the orbital stability for a large
class of NLS also involving the bilaplacian operator as in \cite{BelVi}.
 Indeed, once we know that the weak limit does not
vanish (when $I(u)<0$ for some $u$) by means of the classical
concentration-compactness lemma of Lions \cite{Li}, then Theorem \ref{main-abs}
guarantees that $\bar u$ has the right norm. Finally by the
Cazenave-Lions argument in \cite{CaLi} we have orbital stability.

Finally we show that the abstract Theorem \ref{main-abs} can be applied also to the nonlinear hydrogen atom equation
\begin{equation} \label{pina2}
i\frac{\partial \psi }{\partial t}+\Delta \psi + \frac{1}{|x|}\psi
- \Omega \psi + |\psi|^{p-2}\psi=0
\end{equation}
with $\Omega \in \R$. With the ansatz (\ref{ansatz}) with $\ell \neq 0$, equation (\ref{pina2}) reduces to
\begin{equation} \label{nls-static_hydr}
-\triangle u + \left( \frac{\ell^2}{|y|^2} + \Omega - \frac{1}{|x|} \right) u-u^{p-1}=\lambda u
\end{equation}
where now the potential $V$ depends on $x=(y,z)$ and not only on $y$. The details are given in Section \ref{hydr-sec}.

\section{The elliptic problem} \label{problema-ellittico}
We look for solutions of the equation
\begin{equation*}
-\triangle u + V(|y|)\, u + W'(u)=\lambda u
\end{equation*}
where $x=(y,z)\in \R^k\times \R^{N-k}$, with $N\ge 3$ and $2\le k
\le N$, $\lambda$ is a real parameter, and $V$ and $W$ satisfy the
following assumptions:
\begin{itemize}
\item $V:(0,+\infty) \to \R$ is measurable and
\begin{equation} \label{v-hyp1} \tag{\textbf{V1}}
V(s) \ge 0
\end{equation}
\begin{equation}\label{v-hyp2} \tag{\textbf{V2}}
\lim\limits_{s\to \infty}\, V(s) = 0
\end{equation}
\begin{equation}\label{v-hyp3} \tag{\textbf{V3}}
V(\theta s) \le V(s) \quad \forall\, \theta \ge 1
\end{equation}

\item $W:\R \to \R$ is even and of class $C^{1}$, and writing
$W(s) = \frac{\Omega}{2}\, s^{2} + R(s)$ with $\Omega\in \R$, the
following assumptions hold:
\begin{equation} \label{hyp1} \tag{\textbf{W1}}
R(s)>-b_1s^2-b_2s^{\gamma}, \text{ for some } b_1,b_2>0 \text{ and
} \gamma<2+\frac{4}{N}
\end{equation}
\begin{equation}\label{hyp2} \tag{\textbf{W2}}
|R'(s)| \leq c_1|s|^{q_1-1}+c_2|s|^{q_2-1}, \text{ for some }
c_1,c_2>0 \text{ and } 2 \le q_1 \leq q_2<\frac{2N}{N-2}
\end{equation}
\begin{equation}\label{hyp3} \tag{\textbf{W3}}
\text{ there exists $s_0\in \R^{+}$ such that } W(s_0)<0
\end{equation}
\end{itemize}

Assumptions on $V$ are very general. In particular we don't
require any regularity or boundedness. Singular potentials of the
form $|y|^{-\alpha}$, with $\alpha >0$, are a typical example to
which we are interested. Notice that (\ref{v-hyp1})
follows from (\ref{v-hyp2}) and (\ref{v-hyp3}), and it is explicitly
stated for simplicity.  
Assumptions on $W$ are classical after the paper
\cite{Beres-Lions}. Assumption (\ref{hyp1}) is necessary only to
have that $J_{\rho} > -\infty$ for $\rho$ big enough. Assumption
(\ref{hyp2}) is fundamental to show the existence of the minimum.
In this setting, we remark that we are working with a $C^{1}$
functional which is not weakly semi-continuous and with a non-compact
constraint. We will obtain the sufficient conditions for strong convergence only
for minimizing sequences. Finally assumption (\ref{hyp3}) is natural as it is
necessary for the existence of ground states for the elliptic
equation (\ref{ellittica}) with $V\equiv 0$.

To solve (\ref{ellittica}), we study the minimization problem of
the functional
\begin{equation} \label{func}
J(u):=\int_{\R^N} \left( \frac 12 |\nabla u|^2 +\frac 12 V(|y|)\,
u^2+W(u) \right)dx
\end{equation}
constrained to the manifold
$$
B_{\rho} := \set{ \int_{\R^N} u^2\, dx= \rho^2}
$$
By standard arguments, since $W$ is even we can consider only
non-negative solutions $u$, and since the functional $J$ is
invariant under the action of the group $O(k)$ of orthogonal
transformations on the first $k$ variables of $x \in \R^{N}$, we
can restrict the constraint $B_{\rho}$ to cylindrically symmetric
functions of the form $u=u(|y|,z)$. Let $\OO$ be an open subset of
$\R^{N-k}$, we use the notation $\tilde H^{1}(\R^{k}\times \OO) $
for the Hilbert space obtained as closure of
$C^{\infty}_{0}((\R^{k}\setminus \set{0})\times \OO)$ with respect
to the norm
\begin{equation} \label{norma-cilind}
\| u\|_{H}^2:=\int_{\R^{k}\times \OO}\ \left( |\nabla
u|^2+V(|y|)\, |u|^2+|u|^2 \right)dx
\end{equation}
Moreover we introduce the notation for the ``cylindrical'' part of the norm
\begin{equation} \label{norma-cilind-2}
\| u\|_{c}^2:= \int_{\R^{k}\times \OO}\ \left( |\nabla
u|^2+V(|y|)\, |u|^2 \right)dx
\end{equation}
The subspace of $\tilde H^{1}(\R^{k}\times \OO) $ of cylindrically
symmetric functions will be denoted by $H(\R^{k}\times \OO)$, and
simply by $H$ when $\OO= \R^{N-k}$, hence
$$
H := \set{u\in \tilde H^{1}(\R^{k}\times \R^{N-k})\ :\ u= u(|y|,z)}
$$
Notice that $H \subset \tilde H^{1}(\R^{k}\times \R^{N-k}) \subset H^{1}(\R^{N})$, hence we can use classical Sobolev estimates.

We now restrict the action of the functional $J$ to $H$ and define
\begin{equation}\label{minicil}
J_{\rho}=\inf_{H\cap B_{\rho}} \ J(u)
\end{equation}

Our main results is

\begin{thm} \label{mainth}
If (\ref{v-hyp1})-(\ref{v-hyp3}) and (\ref{hyp1})-(\ref{hyp3})
hold, then for $\rho$ big enough the infimum $J_{\rho}$ defined in
(\ref{minicil}) is achieved.
\end{thm}

Under assumptions (\ref{v-hyp1}) and (\ref{hyp2}), the functional
$J$ in (\ref{func}) is of class $C^{1}$ on $H$, and its critical
points constrained to $B_{\rho}$ satisfy (\ref{ellittica}) for
some $\lambda \in \R$, which is the Lagrange multiplier. Hence as
a corollary of Theorem \ref{mainth} and of the Palais principle of symmetric criticality we get

\begin{thm} \label{wsol-ell}
If (\ref{v-hyp1})-(\ref{v-hyp3}) and (\ref{hyp1})-(\ref{hyp3})
hold, then for $\rho$ big enough equation (\ref{ellittica}) admits
non-negative weak solutions $u$ of $L^{2}$-norm equal to $\rho$.
\end{thm}

By Theorem \ref{wsol-ell}, a solution $u$ satisfies
\begin{equation} \label{form-sol-deb}
\int_{\R^N} \left( \nabla u \cdot \nabla v + V(|y|)\, u\, v+
W'(u)\, v -\lambda\, u\, v\, \right)dx = 0 \qquad \forall\, v \in
\tilde H^{1}(\R^{k}\times \R^{N-k})
\end{equation}
However, we now prove that the point of minimum $u$ satisfies
(\ref{form-sol-deb}) also for all $\phi \in
C^{\infty}_{0}(\R^{N})$, hence it is a solution of
(\ref{ellittica}) in the sense of distributions.

\begin{thm} \label{sol-distrib}
If (\ref{v-hyp1}) and (\ref{hyp2}) hold and $u\in H$ is a
non-negative weak solution of (\ref{ellittica}), then it is a
solution also in the sense of distributions, that is
\begin{equation} \label{form-sol-dist}
\int_{\R^N} \left( \nabla u \cdot \nabla \phi + V(|y|)\, u\, \phi+
W'(u)\, \phi -\lambda\, u\, \phi\, \right)dx = 0 \qquad \forall\,
\phi \in C^{\infty}_{0}(\R^{N})
\end{equation}
\end{thm}

\begin{proof}
We consider the sequence of $C^{\infty}$ non-negative functions
$$
0\le \eta_{n} \le 1, \qquad \eta_{n}(y,z) = \left\{
\begin{array}{ll}
1 & \text{if }\ |y|\ge \frac 2 n \\[0.2cm]
0 & \text{if }\ |y|\le \frac 1 n
\end{array} \right. \qquad \text{and}\ \ |\nabla \eta_{n}| \le K\, n
$$
for a positive constant $K$. Moreover we assume that $\eta_n$ is
non-decreasing along radii starting from $\set{y=0}$. Then for any
$\phi \in C^{\infty}_{0}(\R^{N})$ we have $\eta_{n} \phi \in
\tilde H^{1}$, hence we can choose $v= \eta_{n} \phi$ in
(\ref{form-sol-deb}). Let us assume $\phi \ge 0$. Otherwise we let
$\phi = \phi^{+} - \phi^{-}$ and show (\ref{form-sol-dist})
separately for $\phi^{+}$ and $\phi^{-}$.

Notice that since $u$, $\eta_{n}$ and $\phi$ are non-negative, the
sequence $\set{u\, \phi\, \eta_{n}}$ is non-decreasing and
non-negative, and it converges almost everywhere to $u\, \phi$.
Moreover, since $u\in H$ and $W$ satisfies (\ref{hyp2}), by
classical Sobolev estimates, we get $u\phi \in L^{1}(\R^{N})$ and
$|W'(u)|\phi \in L^{1}(\R^{N})$. Hence, since $\eta_{n}\le 1$, we
can apply Lebesgue dominated convergence theorem to obtain
$$
\int_{\R^N} W'(u)\, \phi \, \eta_{n} \ dx \to \int_{\R^N} W'(u)\, \phi \ dx
$$
$$
\int_{\R^N} u\, \phi \, \eta_{n} \ dx \to \int_{\R^N} u\, \phi \ dx
$$
We write $\nabla u \cdot \nabla (\phi \eta_{n}) = ( \nabla u \cdot
\nabla \eta_{n})\, \phi + (\nabla u \cdot \nabla \phi)\,
\eta_{n}$, and since $u \in H$ we have $|\nabla u \cdot \nabla
\phi| \in L^{1}(\R^{N})$. Hence as above
$$
\int_{\R^N} (\nabla u \cdot \nabla \phi)\, \eta_{n}\ dx \to
\int_{\R^N}  \nabla u \cdot \nabla \phi \ dx
$$
Moreover, letting $A_{n} = \set{|y| \le \frac 2 n} \cap \text{supp
$\phi$}$, it holds $m(A_{n}) \le const.\ \frac{1}{n^{k}}$, hence
$$
\int_{\R^N} |\nabla u \cdot \nabla \eta_{n}|\, \phi\ dx \le
const\, \frac{1}{n^{\frac k 2 -1}}\, \| \phi\|_{\infty}\, \left(
\int_{A_{n}} |\nabla u|^{2}\ dx \right)^{\frac 1 2} = o(1)
$$
for all $k\ge 2$ since $u \in H$ and $m(A_{n})\to 0$.

Writing (\ref{form-sol-deb}) with $v= \eta_{n} \phi$, using
previous results we get
\begin{equation} \label{grande-badiale}
\lim_{n\to \infty}\ \int_{\R^N} V(|y|)\, u\, \phi\, \eta_{n} \ dx
= \int_{\R^N} \left(  \nabla u \cdot \nabla \phi + W'(u)\, \phi -
\lambda\, u\, \phi \right)\ dx \in \R
\end{equation}
Since the sequence $\set{V(|y|)\, u\, \phi\, \eta_{n}}$ is
non-decreasing and non-negative, we get
$$
\int_{\R^N} V(|y|)\, u\, \phi\, \ dx = \lim_{n\to \infty}\
\int_{\R^N} V(|y|)\, u\, \phi\, \eta_{n} \ dx
$$
which together with (\ref{grande-badiale}) implies (\ref{form-sol-dist}).
\end{proof}

\section{Proof of Theorem \ref{mainth}} \label{proof}

We first prove an abstract result. Consider the
minimization problem
\begin{equation} \label{mini1}
I_{\rho}=\inf_{H \cap B_{\rho}} \ I(u)
\end{equation}
\begin{equation} \label{func1}
I(u):=\left( \frac 12 \| u \|^2_{c} + T(u) \right)
\end{equation}
where 
\begin{equation}\nonumber
\| u\|_{c}^2:= \int_{\R^{k}\times \OO}\ \left( |\nabla
u|^2+V(|y|)\, |u|^2 \right)dx
\end{equation}
with $V$ satisfying (\ref{v-hyp1})-(\ref{v-hyp3})
and $T$ is a real operator on $H$. Then

\begin{thm} \label{main-abs}
Let $T$ be differentiable on $H$ and $\set{u_n}\subset H\cap B_{\rho}$ be a minimizing sequence for (\ref{mini1}). Assume also that
\begin{equation}\label{(1)}
 u_n\rightharpoonup \bar u\neq 0 ;
\end{equation}
 \begin{equation}\label{(2)}
T(u_n-\bar u) + T(\bar u)=T(u_n)+ o(1);
\end{equation}
\begin{equation}\label{(4)}
T(\alpha_{n}(u_n-\bar u)) -T(u_n-\bar u) = o(1) \qquad \forall\, \set{\alpha_{n}}\subset \R\ \text{s.t.}\ \alpha_{n} \to 1;
\end{equation}
\begin{equation} \label{(6)}
<T'(u_{n}), u_{n}> = O(1)
\end{equation}
\begin{equation} \label{(5)}
<T'(u_{n}) - T'(u_{m}), u_{n} - u_{m}> = o(1) \qquad \text{as}\ n,m \to \infty
\end{equation}
\begin{equation}\label{(3)}
 \ T(u_{\theta}) \leq \theta^2 T(u)  \qquad \forall\,u\in H
\end{equation}
where
$$
u_{\theta}(x):=u\left(\frac{x}{\theta^{\frac{2}{N}}}\right), \theta>1
$$
Then $\bar u \in B_{\rho}$ and, up to a sub-sequence, $\| u_{n} - \bar u \|_{H} \to 0$ (see (\ref{norma-cilind})).
\end{thm}

\begin{proof}
By (\ref{(1)}) we have $\| \bar u\|_{L^2} = \mu \in (0,\rho]$ and we assume that $\mu < \rho$, then we obtain a contradiction. Notice that again by (\ref{(1)}) we have
$$
\| u_n-\bar u\|_{L^2}^2+\|\bar u\|_{L^2}^2=\| u_n\|_{L^2}^2+o(1)
$$
hence
\begin{equation} \label{convergenza-norme}
\alpha_{n} := \frac{\sqrt{\rho^2-\mu^2}}{\| u_n-\bar u\|_{L^2}} \to 1
\end{equation}
By definition (\ref{mini1})
$$
\frac 12 \| u_{n} \|^2_{c} + T(u_{n}) =I_{\rho}+o(1)
$$
and by (\ref{(2)})
$$
\frac 12 \| u_n-\bar u\|_{c}^2 +\frac 12 \|\bar u\|_{c}^2+T(u_n-\bar u)+T(\bar u)=I_{\rho}+o(1)
$$
Hence, using the sequence $\alpha_{n}$ defined in (\ref{convergenza-norme}), by (\ref{(4)})
$$
\frac 12 \| \alpha_{n}(u_n-\bar u) \|_{c}^2 + T(\alpha_{n}(u_n-\bar u)) + \frac 12 \|\bar u\|_{c}^2+T(\bar u)=I_{\rho}+o(1)
$$

For $u_{\theta}$ we have $\|u_{\theta}\|_{L^2}=\theta\rho$ and 
\begin{eqnarray}
||u_{\theta}||_c^2= \left(\int_{\R^N} \frac{\theta^2}{\theta^{4/N}}
|\nabla u|^2 +  \theta^2\, V(|\theta y|)\, u^2 dx )\right)< \nonumber \\
<  \theta^2 \left(\int_{\R^N}|\nabla u|^2+
V(|y|)\, u^2 dx )\right)=\theta^2\|u(x)\|^2_c. \nonumber
\end{eqnarray}
By (\ref{(3)}) 
\begin{equation}\nonumber
I_{\theta \rho}= \inf \left(\frac 12 \|u_{\theta}|\|_{c}^2+T(u_{\theta}) \right) 
<  \inf \theta^2 \left(\frac 12 \|u|\|_{c}^2+T(u)\right)= \theta^2I_{\rho}
\end{equation}
and thus 
\begin{equation}\label{sub}
I_{\rho}<I_{\mu}+I_{\sqrt{\rho^2-\mu^2}}
\end{equation}
for any $\rho>0$
and $\theta>1$.\\
Now, notice that $\| \alpha_{n}(u_n-\bar u) \|_{L^{2}} = \sqrt{\rho^{2}-\mu^{2}}$, hence
$$
I_{_{\sqrt{\rho^{2}-\mu^{2}}}} + I_{\mu} \le \frac 12 \| \alpha_{n}(u_n-\bar u) \|_{c}^2 + T(\alpha_{n}(u_n-\bar u)) + \frac 12 \|\bar u\|_{c}^2+T(\bar u) = I_{\rho} + o(1)
$$
which is in contradiction with (\ref{sub}). This implies that $\| \bar u \|_{L^{2}} = \rho$.

>From $\bar u \in B_{\rho}$ it follows that $\| u_n-\bar u
\|_{L^{2}} = o(1)$, hence it remains to show that $\| u_{n} - \bar
u \|_{c} = o(1)$ up to a sub-sequence. By Ekeland principle, we
can assume that there exists a sequence $\set{\lambda_{n}}\subset
\R$ such that for the functional $I$ defined in (\ref{func1})
$$
<I'(u_{n}) - \lambda_{n}\, u_{n} , v > = o(1) \qquad \forall\, v \in H
$$
where $<\cdot,\cdot>$ denotes the duality pairing. It follows that
$$
<I'(u_{n}) - \lambda_{n}\, u_{n} , u_{n} > = o(1)
$$
since $\| u_{n} \|_{H}$ is bounded. From this and assumption (\ref{(6)}) it follows that the sequence $\set{\lambda_{n}}$ is bounded, hence up to a sub-sequence there exists $\lambda \in \R$ with $\lambda_{n} \to \lambda$.

We now have
$$
<I'(u_{n}) - I'(u_{m}) - \lambda_{n}u_{n} + \lambda_{m} u_{m}\ ,\ u_{n} - u_{m}>  = o(1) \qquad \text{as}\ n,m\to \infty
$$
hence, using $(\lambda_{n} - \lambda_{m}) <u_{m}, u_{n} - u_{m}> = o(1)$,
$$
\| u_{n} - u_{m} \|_{c}^{2} + <T'(u_{n}) - T'(u_{m}), u_{n} - u_{m}> - \lambda_{n} \| u_{n} - u_{m}\|_{L^{2}}^{2} = o(1)
$$
Since $\| u_{n} - u_{m}\|_{L^{2}} = o(1)$, $\lambda_{n} \to
\lambda$ and (\ref{(5)}) holds, we obtain that $\set{u_{n}}$ is a
Cauchy sequence in $H$. Hence $\|u_{n} - \bar u\|_{H} \to 0$.
\end{proof}

\begin{remark} \label{ipotesi-t}
Notice that (\ref{(2)})-(\ref{(5)}) in the previous theorem are assumed only for minimizing sequences. Moreover (\ref{(4)}) holds for example for uniformly continuous operators $T$.
\end{remark}

The proof of Theorem \ref{mainth} is now reduced to show that assumptions of Theorem \ref{main-abs} are satisfied for $J$ defined in (\ref{func}), with $T(u) = \int W(u)$. This is obtained by the following lemmas.

\begin{lem} \label{stimprio}
If (\ref{hyp1}) holds then $J_{\rho}> - \infty$, and any minimizing sequence $\set{u_n}\subset H\cap B_{\rho}$, i.e. $J(u_n)\rightarrow J_{\rho}$, is bounded in $H$.
\end{lem}

\begin{proof}
We apply the Sobolev inequality (see \cite{VoPa})
\begin{equation}\label{Sobolev}
\| u \|_{L^q} \leq b_{q} \| u \|_{L^2}^{1-\frac{N}{2}+\frac{N}{q}}\, \| \nabla
u \|_{L^2}^{\frac{N}{2}-\frac{N}{q}}
\end{equation}
that holds for $2 \leq q \leq 2^{*}$ when $N \geq 3$. From
(\ref{Sobolev}) it follows that for any $u \in B_{\rho}$
\begin{equation}\label{sopra}
\|u \|_{L^q}^q \leq b_{q,\rho} \|\nabla u\|_{L^2}^{\frac{qN}{2}-N}.
\end{equation}
Now, by (\ref{sopra}) and (\ref{hyp1}), for all $u\in H \cap B_{\rho}$
\begin{eqnarray}
J(u)&\geq& \int \left( \frac 12 |\nabla u|^2 +\frac 12 V(|y|) u^2 + \frac{\Omega}{2} u^{2} -b_1u^2-b_2u^{\gamma} \right)\, dx \nonumber\\
&\geq &\int \left( \frac 12 |\nabla u|^2 + \frac 12 V(|y|) u^2
\right) dx- b_2 b_{\gamma,\rho}\left(\int{|\nabla
u|^2}dx\right)^{\frac{\gamma N}{4}-\frac{N}{2}} + \left(
\frac{\Omega}{2} - c_1\right) \rho^2\nonumber.
\end{eqnarray}
Since $\gamma < 2+\frac{4}{N}$, it holds $\frac{\gamma N}{2}-N<2$, hence
have
$$
J(u)\geq \frac{1}{2}\| u \|_{c}^2 +o \left( \| \nabla u \|_{L^2}^2 \right).
$$
The proof follows easily.
\end{proof}

\begin{lem}\label{min0}
If (\ref{v-hyp2}) and (\ref{hyp3}) hold, then there exists
$\rho_0$ such that $J_{\rho}<0$ for all $\rho>\rho_0$.
\end{lem}
\begin{proof}
Let $s_0$ satisfy $W(s_0)<0$ as in (\ref{hyp3}). We consider the
sequence of functions $u_n(x) = u_n(|y|,z) = f(|z|)\, v_n(|y|)$
with
$$
v_n(|y|) = \left\{
\begin{array}{ll}
s_0 (|y|-R_n+1) & \text{for $R_n-1 \le |y| \le R_n$} \\[0.1cm]
s_0 & \text{for $R_n \le |y| \le 2 R_n$} \\[0.1cm]
s_0 (2 R_n-|y|+1) & \text{for $2 R_n \le |y| \le 2 R_n+1$} \\[0.1cm]
0 & \text{for $|y| \ge 2 R_n+1$}
\end{array} \right.
$$
$$
f(|z|) = \left\{
\begin{array}{ll}
1 & \text{for $0\le |z| \le 1$} \\[0.1cm]
2-|z| & \text{for $1 \le |y| \le 2$} \\[0.1cm]
0 & \text{for $|y| \ge 2$}
\end{array}
\right.
$$
and assume $R_n \to \infty$. Then $u_n \in H$ and
$$
\int\, |\nabla u_n|^2\, dx = O(R_n^{k-1})
$$
$$
\int\, V(|y|)\, u_n^2\, dx = \left( \int_{R_n}^{2 R_n}\, V(r)\,
r^{k-1} \, s_0^2\, dr \right) \, \left( \int_{|z|\le 2}\, f(z) dz
\right) + o(R_n^{k}) = O(R_n^{k}\, V(R_n))
$$
$$
\int\, W(u_n)\, dx = \left( \int_{R_n}^{2 R_n}\, r^{k-1}\,
W(s_0)\, dr \right) \, \left( \int_{|z|\le 2}\, f(z) dz \right) +
o(R_n^{k-1}) = O(R_n^k)
$$
Since $W(s_0)<0$ and $V(R_n)\to 0$, it follows that $J(u_n)$ is
negative for $n$ large enough, and $\| u_n \|_{L^2}^2 = O(R_n^k)$.
\end{proof}

\begin{lem} \label{compactness}
Let $2\le k < N$ and $\OO$ a bounded open subset of $\R^{N-k}$.
Then the embedding $H(\R^k \times \OO) \hookrightarrow L^p(\R^N)$
is compact for all $p\in (2,2^{*})$. If $k=N$ then the embedding
$H(\R^N) \hookrightarrow L^p(\R^N)$ is compact for all $p\in
(2,2^{*})$.
\end{lem}

\begin{proof}
It follows from the compactness results in \cite{ESLI83} for $2\le
k < N$ and in \cite{St77} for $k=N$.
\end{proof}

\begin{lem}\label{van}
Let $J_{\rho}<0$ and $u_n$ be a minimizing sequence for
(\ref{minicil}) under assumptions (\ref{hyp2}). Then, up to
translations in $\R^{N-k}$, we have $u_n \rightharpoonup \bar
u\neq 0$.
\end{lem}
\begin{proof}
Since $J_{\rho}<0$, by (\ref{hyp2}) there exists $q\in (2,2^{*})$
such that
$$
\int_{\R^n}u_n^q\, dx\geq \beta>0.
$$
Moreover, by Lemma \ref{stimprio}, $\| u_n \|_H$ is bounded and
there exists $\bar u \in H$ such that $u_n \rightharpoonup \bar
u$. It remains to prove that $\bar u \not\equiv 0$.

Now we introduce for every $j\equiv(j_{k+1},j_{k+2},...,j_N)\in
\Z^{N-k}$ the cube
$$Q_j\equiv [j_{k+1},j_{k+1}+1)\times [j_{k+2},j_{k+2}+1)\times...\times [j_{N},j_{N}+1) \in \R^{N-k}$$
Let $S_j\equiv \R^{k}\times Q_j$, we have
\begin{eqnarray}
&&0<\beta \leq \int_{\R^N} u_n^q\, dx=\sum_j \left( \int_{S_j} |u_n|^q\, dydz\right)= \nonumber \\
&&=\sum_j \left( \int_{S_j} |u_n|^{q-2}|u_n|^2dydz\right)\leq \sum_j \left(\int_{S_j} |u_n|^{q}dydz \right)^{\frac{q-2}{q}}\left(\int_{S_j} |u_n|^{q}dydz \right)^{\frac{2}{q}}\leq  \nonumber \\
&&\leq  \sup_j\left(\int_{S_j} |u_n|^{q}dydz \right)\left(\sum_j\int_{S_j} |u_n|^{q}dydz \right)^{\frac{2}{q}}\leq \nonumber \\
&&\leq const\  \sup_j\left(\int_{S_j} |u_n|^{q}dydz \right) \
||u_n||_{H}^{\frac 2 q} \nonumber
\end{eqnarray}
Hence there exists a sequence of cubes $Q_{j^n}$ such that
$$\int_{S_{j^n}}\, |u_n|^{q}dydz>const>0$$
It follows that the minimizing sequence $v_n(x):=u_n(x+j^n)$
satisfies
$$
\int_{\R^k\times Q_0}\, |v_n|^q\, dx > const >0
$$
hence, by Lemma \ref{compactness}, the weak limit $\bar u
\not\equiv 0$.
\end{proof}

\begin{lem} \label{brezis-lieb}
Let $(u_n) \in L^{q_1}(\R^N)\cap L^{q_2}(\R^N)$, with $1< q_1 \leq q_2< \infty$, and $\tilde T(u)=\int R(u)dx$ with $R:\R\rightarrow \R$ of class $C^1$.  If
\begin{itemize}
\item $(u_n)$ is bounded in $L^{q_1}(\R^n)\cap L^{q_2}(\R^n)$;
\item $u_n \rightarrow u$ almost everywhere;
\item $|R'(s)| \leq b_1|s|^{q_1-1}+b_2|s|^{q_2-1}$,  $1 < q_1 \leq q_2<\infty$
\end{itemize}
Then
$$
\tilde T(u_n-u)+ \tilde T (u)= \tilde T (u_n)+o(1).
$$
\end{lem}

\begin{proof}
First of all we can write
$$
\left| R(u_{n}) - R(u_{n}-u) \right| = \left| R'(u_{n}-u+ \theta u) \right|\, |u| \le
$$
$$
\le c_1|u_{n}-u|^{q_1-1} \, |u| + c_2 |u_{n}-u|^{q_2-1}\, |u| + c_{3} (|u|^{q_{1}}+ |u|^{q_{2}})
$$
for some $\theta \in (0,1)$. Moreover, applying for any fixed $\epsilon>0$ the Young inequality
$$
ab \leq \epsilon a^p+c(\epsilon)b^q   \qquad \text{with}\ a,b>0\ \text{and} \ \frac{1}{p}+\frac{1}{q}=1
$$
with coniugated exponents $\frac{q_1}{q_1-1}, q_1$ and $\frac{q_2}{q_2-1}, q_2$ we get
$$
|R(u_n)-R(u_n-u)-R(u)|\leq \epsilon c_{4} (|u_n-u|^{q_1}+|u_n-u|^{q_2})+c(\epsilon) c_{5}(|u|^{q_1}+|u|^{q_2})+|R(u)|
$$
Hence, setting
$$
f_n^{\epsilon}:=|R(u_n)-R(u_n-u)-R(u)|-\epsilon c_{4}(|u_n-u|^{q_1}+|u_n-u|^{q_2})
$$
the Lebesgue dominated convergence theorem implies
$$
\lim_{n\rightarrow \infty}\int f_n^{\epsilon}dx=0.
$$
The proof is finished by writing
$$
|\tilde T(u_n)-\tilde T(u_n-u)-\tilde T(u)|\leq \int |R(u_n)-R(u_n-u)-R(u)|dx \leq
$$
$$
\leq \int f_n^{\epsilon}dx+\epsilon c_{4} \left(\int( |u_n-u|^{q_1}+|u_n-u|^{q_2})dx \right) \le o(1) + \epsilon c_{4} K
$$
where $K:=sup_{n}|\int (|u_n-u|^{q_1}+|u_n-u|^{q_2}|)dx$.
\end{proof}

\begin{lem} \label{t-ass}
If (\ref{v-hyp2}) and (\ref{hyp2}) hold, then the operator $T(u) =
\int W(u)$ satisfies (\ref{(2)})-(\ref{(5)}) for any minimizing
sequence $\set{u_{n}} \subset H\cap B_{\rho}$, for $\rho$ large
enough.
\end{lem}

\begin{proof}
By Lemma \ref{stimprio}, any minimizing sequence is bounded in the
$H$ norm. Hence $\set{u_n}$ is bounded in all $L^p$ norms for $p
\in [2,2^*]$ and there exists $\bar u \in H$ such that $u_n
\rightharpoonup \bar u$. Moreover, classical compact embeddings of
$H$ into $L^p$ when restricted to open bounded subsets of $\R^N$
imply almost everywhere convergence of $u_n$ to $\bar u$.

Writing $T(u) = \int W(s) = \int \frac{\Omega}{2}\, u^2 + \tilde
T(u)$, condition (\ref{(2)}) is satisfied for $\int
\frac{\Omega}{2}\, u^2$ by a standard argument, and for $\tilde
T(u)$ by (\ref{hyp2}) and Lemma \ref{brezis-lieb}.

Condition (\ref{(4)}) is immediate for $\int \frac{\Omega}{2}\,
u^2$. For $\tilde T(u)$ we write
$$
\left| \int\, \left( R(\alpha_n(u_n-\bar u)) - R(u_n-\bar u)
\right)\, dx \right| \le |\alpha_n-1|\ \int\, |R'((1+\theta)
(u_n-\bar u))|\, |u_n-\bar u|\, dx
$$
for some $\theta\in (0,1)$. Hence, by (\ref{hyp2}),
$$
|\tilde T(\alpha_n(u_n-\bar u)) - \tilde T(u_n-\bar u)| \le
const\, |\alpha_n-1|\ \max\set{\| u_n - \bar u\|_{L^q}^q, \| u_n -
\bar u\|_{L^p}^p }
$$
Condition (\ref{(6)}) follows using (\ref{hyp2}) as above, and by
boundedness of $\| u_n - \bar u\|_{L^p}$ for all $p\in [2,2^*]$.

It remains to prove (\ref{(5)}). By Lemmas \ref{min0} and
\ref{van}, we obtain $\bar u \not\equiv 0$ if the minimization
problem (\ref{minicil}) is studied in $H\cap B_\rho$ for $\rho$
large enough. Hence we can repeat the proof of Theorem
\ref{main-abs} to obtain that the weak limit $\bar u$ is in
$B_\rho$. From $\bar u \in B_{\rho}$ it follows that $\| u_n-\bar
u \|_{L^{2}} = o(1)$. This implies (\ref{(5)}) for the term $\int
\frac{\Omega}{2}\, u^2$.

We now prove this condition for $\tilde T(u)$. We recall that $\| u_n-\bar u
\|_{H} = O(1)$. Now, the following inequality for $v \in H$
$$
\int_{\R^{N}}\, |v|^{p} \le \left( \int_{\R^{N}}\, |v|^{2}
\right)^{\frac{p-\beta}{2}}\, \left( \int_{\R^{N}}\, |v|^{2^{*}}
\right)^{\frac{\beta}{2^{*}}}
$$
with $2^{*}= \frac{2N}{N-2}$, $p\in (2,2^{*})$ and $\beta= \frac N
2 (p-2) \in (0,p)$, and the classical Sobolev theorem imply
\begin{equation} \label{norma-p}
\| u_{n} - \bar u \|_{L^{p}}^{p} \le const\, \| u_n-\bar u
\|_{L^{2}}^{p-\beta}\, \| u_n-\bar u \|_{H}^{\beta} = o(1)
\end{equation}
for all $p\in (2,2^{*})$. The proof of (\ref{(5)}) is now finished
by writing
$$
\left| \int\, (R'(u_{n}) - R'(u_{m}))(u_{n}-u_{m})\, dx \right| \le \int\, (|R'(u_{n})| + |R'(u_{m})|)\, |u_{n}-u_{m}|\, dx
$$
Using (\ref{hyp2}) and
$$
\int\, |u_{n}|^{p-1}\, |u_{n}-u_{m}|\, dx \le \left( \int\, |u_{n}|^{p}\, dx \right)^{\frac 1 q}\ \left( \int\, |u_{n}-u_{m}|^{p}\, dx \right)^{\frac 1 p}
$$
which holds for all $p\in (2,2^{*})$, we get
$$
\left| \int\, (R'(u_{n}) - R'(u_{m}))(u_{n}-u_{m})\, dx \right| \le const\ \left( \| u_{n}-u_{m} \|_{L^{p}} + \| u_{n}-u_{m} \|_{L^{q}} \right)
$$
where $p,q$ are as in (\ref{hyp2}). Condition (\ref{(5)}) now follows from (\ref{norma-p}).
\end{proof}

\begin{lem} \label{convexity}
Let $T(u)=\int_{\R^N}W(u)dx$. For any $\theta>1$ and $u \in H$ we have
$ T(u_{\theta})= \theta^2T(u)$.
\end{lem}

\vskip 0.3cm

\begin{proof}[Proof of Theorem \ref{mainth}]
If (\ref{v-hyp1})-(\ref{v-hyp3}) and (\ref{hyp1})-(\ref{hyp3})
hold, by Lemmas \ref{stimprio}, \ref{min0}, \ref{van}, \ref{t-ass}
and \ref{convexity}, we can apply Theorem \ref{main-abs} to the
functional $J$ defined in (\ref{func}). Hence for any minimizing
sequence $\set{u_n} \subset H \cap B_\rho$ with $\rho$ large
enough, there exists $\bar u \in B_\rho$ such that $\| u_n - \bar
u \|_H \to 0$. Now, from the continuity of the functional $J$ it
follows that $J(\bar u) = J_\rho$.
\end{proof}

\section{Application to nonlinear hydrogen atom} \label{hydr-sec}
For the nonlinear hydrogen atom equation (\ref{pina2}), we prove the existence of solutions of the elliptic equation 
\begin{equation} \label{nls-static_hydr2}
-\triangle u + \left( \frac{\ell^2}{|y|^2}+\Omega - \frac{1}{|x|} \right)
u-u^{p-1}=\lambda u
\end{equation}
The solutions of (\ref{nls-static_hydr2}) are critical points of the functional
$$
G(u):=\int_{\R^N} \left( \frac 12 |\nabla u|^2 
+\frac 12 \frac{\ell^2}{|y|^2}u^2+\frac{\Omega}{2}u^2-\frac 12 \frac{u^2}{|x|}-\frac{u^p}{p}\right)dx$$
restricted to cylindrically symmetric functions and constrained to
the manifold
$$
B_{\rho} := \set{ \int_{\R^N} u^2\, dx= \rho^2}
$$
We may assume $\Omega >1$. Indeed if $\Omega \le 1$, we first look for solutions of the elliptic equation
$$
-\triangle u + \left( \frac{\ell^2}{|y|^2}+ 2 - \frac{1}{|x|} \right)
u-u^{p-1}=\lambda u
$$
Then given such a solution $u$, the function
$$
\psi(t,x) = u(x) e^{i(\ell \theta(y) - (\lambda -2 + \Omega) t)}
$$
is a solution of (\ref{pina2}).

We first notice that
\begin{lem} \label{e-norma}
For $\Omega >1$ and $\ell\not= 0$ it holds
$$
V(x) := \frac{\ell^2}{|y|^2}+\Omega-\frac{1}{|x|} \ge 0  
$$
for all $x\in \R^{N}$.
\end{lem}
Hence we can use the setting introduced in Section \ref{problema-ellittico} with $k=2$. We introduce $\tilde H^{1}(\R^{2} \times \OO)$ as the closure of $C^{\infty}_{0}((\R^{2}\setminus \set{0}) \times \OO)$ with respect to the norm
\begin{equation} \label{norma-cilind2}
\| u\|^2:=\int_{\R^{2}\times \OO}\ \left( |\nabla
u|^2+ V(x)u^2 \right) dx
\end{equation}
The subspace of $\tilde H^{1}(\R^{2}\times \OO) $ of cylindrically
symmetric functions will be denoted by $H(\R^{2}\times \OO)$, and
simply by $H$ when $\OO= \R^{N-2}$, hence
$$
H := \set{u\in \tilde H^{1}(\R^{2}\times \R^{N-2})\ :\ u= u(|y|,z)}
$$
We now restrict the action of the functional $G$ to $H$ and define
\begin{equation} \label{inf-g}
G_{\rho}=\inf_{H\cap B_{\rho}} \ G(u)
\end{equation}
We prove the following theorem

\begin{thm}\label{main2}
Let $N\geq 3$ and $2<p<2+\frac{4}{N}$. Then for $\rho$ big enough the infimum (\ref{inf-g}) of $G$ is achieved. Hence there are radially symmetric solutions (in the sense of distributions) of (\ref{nls-static_hydr2}) with prescribed $L^{2}$ norm large enough.
\end{thm}

We now show how to obtain the existence of the point of minimum. That this implies the existence of weak solutions of (\ref{nls-static_hydr2}) is immediate. That these solutions are also solutions in the sense of distributions follows as in Theorem \ref{sol-distrib}.

The proof follows as in Theorem \ref{main-abs} and the continuity of $G$ in $H$. Hence we need to show that $G$ satisfies assumptions  (\ref{(1)})-(\ref{(5)}).
Moreover assumption (\ref{(3)}) is replaced by

\begin{lem} \label{convexity2}
For any $\mu \in (0,\rho)$ it holds
$G_{\rho}<G_{\mu}+G_{_{\sqrt{\rho^2-\mu^2}}}$.
\end{lem}
\begin{proof}
It is sufficient to prove that $G_{\theta \rho}< \theta^2 G_{\rho}$ for all $\theta>1$. For $u\in H \cap B_{\rho}$ we set $u_\theta(x) \equiv \theta u(x)$. Clearly $\|u_\theta \|_{2}^2=\theta^2\rho^2$. Using the norm (\ref{norma-cilind2}) we get
$$
G_{\theta\rho} \le \inf \left( \theta^2 \| u\|^2 - \frac{\theta^p}{p} \int_{\R^3} |u|^p dx \right) <  \inf \theta^2 \left( \| u\|^2 - \frac 1p  \int_{\R^3} |u|^pdx \right) = \theta^2 G_{ \rho}
$$
 \end{proof}
Following the ideas of Lemmas \ref{stimprio}, \ref{min0} and \ref{van}, we obtain

\begin{lem} \label{stimprio2}
$G_{\rho}> - \infty$, and any minimizing sequence $\set{u_n}\subset H\cap B_{\rho}$, i.e. $G(u_n)\rightarrow G_{\rho}$, is bounded in $H$.
\end{lem}

\begin{lem}\label{neg2}
There exists $\rho_0$ such that $G_{\rho}<0$ for all $\rho>\rho_0$.
\end{lem}
\begin{proof}
It follows as in Lemma \ref{min0} by choosing $s_0$ such that
$\left(\frac {\Omega}{2} s_0^2-\frac{s_0^p}{p}\right)<0$.
\end{proof}
\begin{lem}\label{nonvan2}
Let $G_{\rho}<0$ and $u_n$ be a minimizing sequence. Then, up to
translations in $\R$, we have $u_n \rightharpoonup \bar u\neq 0$.
\end{lem}
\begin{proof}
It follows as in Lemma \ref{van} by choosing simply by noticing
that translations along the $z$ variables make the functional $G$
to decrease. Hence $v_n(x):=u_n(|y|,z+j^n)$ satisfies $G(v_n)\leq G(u_n)$.
\end{proof}

\begin{proof}[Proof of Theorem \ref{main2}]
By Lemmas  \ref{stimprio2}, \ref{neg2} and \ref{nonvan2}, the functional $G$ satisfies (\ref{(1)})  of Theorem \ref{main-abs}. Conditions (\ref{(2)})-(\ref{(5)}) are obtained as in Lemmas \ref{brezis-lieb} and \ref{t-ass}. By Lemma \ref{convexity2} we have the subadditivity condition $I_{\rho}<I_{\mu}+I_{_{\sqrt{\rho^2-\mu^2}}}$ that in Theorem \ref{main-abs} is guaranteed by condition (\ref{(3)}).

Hence we can apply Theorem \ref{main-abs} to $G$. The proof follows from the continuity of $G$.
\end{proof}

\end{document}